\documentclass[letterpaper, 10 pt, conference]{ieeeconf}
\usepackage{graphicx} 

\usepackage{amssymb} 
\usepackage{amsmath}
\usepackage{soul}

\usepackage[dvipsnames]{xcolor}

\usepackage{comment}

\usepackage{subfigure}

\renewcommand{\baselinestretch}{0.99}

\IEEEoverridecommandlockouts 

\newcommand{\outcome}[1]{J(#1)}
\newcommand{\graph}[1]{{G^{#1}}}
\newcommand{\graphset}{\mathcal{G}}
\newcommand{\finaltime}{{T_f}}
\newcommand{\nodeset}{\mathcal{V}}
\newcommand{\edgeset}{\mathcal{E}}
\newcommand{\weightmatrix}[1]{W^{#1}}
\newcommand{\weightmatrixset}{\mathcal{W}}

\newcommand{\blueposition}[1]{p_{#1}}
\newcommand{\ammo}[1]{\alpha(#1)}
\newcommand{\maxammo}{\overline{\alpha}}
\newcommand{\neighbors}[2]{{\overrightarrow{\mathcal{N}}^\text{#1}_{#2}}}
\newcommand{\markovstate}{S}
\newcommand{\markovstateset}{\mathcal{S}}

\newcommand{\numbernodes}{N}
\newcommand{\numbergraphs}{K}
\newcommand{\numberagents}{M}
\newcommand{\stagecost}[1]{C(#1)}

\newcommand{\policyred}{\psi_\text{red}}
\newcommand{\policyblue}{\psi_\text{blue}}

\newcommand{\mixedpolicyblue}{\pi_\text{blue}}
\newcommand{\mixedpolicyred}{\pi_\text{red}}
\newcommand{\securitypolicyblue}{\hat{\pi}_\text{blue}}
\newcommand{\securitypolicyred}{\hat{\pi}_\text{red}}

\newcommand{\actionsetblue}{A_\text{blue}}
\newcommand{\actionsetred}{A_\text{red}}
\newcommand{\transitionprob}{P}

\newcommand{\expectation}[2]{\mathbb{E}_{ #2 }\left[ #1 \right]}

\newcommand{\upperoutcome}[1]{\overline{J}\left( #1 \right)}
\newcommand{\loweroutcome}[1]{\underline{J}\left( #1 \right)}
\newcommand{\valueperagent}[1]{\Tilde{V}_{#1}}

\newcommand{\redGraph}{\mathcal{G}^\text{red}}
\newcommand{\redNodes}{\mathcal{V}^\text{red}}
\newcommand{\redEdges}{\mathcal{E}^\text{red}}

\newcommand{\jointgraph}[1]{{\mathcal{G}'^{#1}}}
\newcommand{\jointvertexset}{\mathcal{V}'}
\newcommand{\jointedges}{\mathcal{E}'}
\newcommand{\jointweightmatrix}[1]{\mathcal{W}'^{#1}}
\newcommand{\jointweightmatrixset}{\overline{\mathcal{W}}'}
\newcommand{\jointgraphset}{\overline{\mathcal{G}}'}
\newcommand{\jointneighbors}[2]{{\mathcal{N}'^{\text{#1}}_{#2}}}

\newcommand{\iter}{\tau}

\newcommand{\qfunction}{Q}
\newcommand{\qfunctionapprox}[1]{\hat{Q}^{#1}}
\newcommand{\qmatrix}{\mathcal{Q}}

\newcommand{\worstcasegraph}{\overline{G}}
\newcommand{\worstcaseweights}{\overline{W}}

\newcommand{\naturalSet}[1]{\left[#1\right]}

\newtheorem{theorem}{Theorem}

\title{
Multi-Robot Coordination Induced in an Adversarial Graph-Traversal Game
}
\author{James Berneburg, Xuan Wang, Xuesu Xiao, and Daigo Shishika
\thanks{The authors are with the Department of Mechanical Engineering, George Mason University, Fairfax, VA 22030, USA,  {\tt\small\{jbernebu,xwang64,xiao,dshishik\}}@gmu.edu}
}

\date{August 2023}

\begin{document}

\maketitle

\begin{abstract}
This paper presents a game theoretic formulation of a graph traversal problem, with applications to robots moving through hazardous environments in the presence of an adversary, as in military and security scenarios. The blue team of robots moves in an environment modeled by a time-varying graph, attempting to reach some goal with minimum cost, while the red team controls how the graph changes to maximize the cost. The problem is formulated as a stochastic game, so that Nash equilibrium strategies can be computed numerically. Bounds are provided for the game value, with a guarantee that it solves the original problem. Numerical simulations demonstrate the results and the effectiveness of this method, particularly showing the benefit of mixing actions for both players, as well as beneficial coordinated behavior, where blue robots split up and/or synchronize to traverse risky edges. 


\end{abstract}

\section{Introduction}


Consider a scenario where multiple robots must traverse difficult terrain to reach a goal in the presence of an adversary, as shown in Figure~\ref{fig:motivatingGraph}. The robots 
must plan a path through the environment to minimize risk, while the adversary affects the risk by altering the environment, such as by destroying a bridge or by positioning its own robots, in military or security scenarios. 

Such a hazardous environment is often modeled as a graph, 
which indicates both the different paths robots may take to reach the goal and the risk of those paths through edge weights. 
For example, probabilistic roadmap planners construct a graph from the environment which is then planned over~\cite{gasparetto2015path,hsu2002randomized}. 
For a changing environment, one may consider finding the shortest (i. e. minimal risk) path through time-varying graphs, such as in~\cite{ding2008finding,cai1997time,yuan2019constrained} which find 
shortest time paths when edges have time-varying delays. 
The work in~\cite{dimmig2023multi} specifically considers planning for a team of robots on a graph, formulating it as an optimization problem 
where robots can support each other. 

Another methodology to handle uncertainty in robotics problems is to model them as Markov decision processes (MDPs)~\cite{puterman1990markov}, assuming a transition function is known. In particular, partially observable MDPs~\cite{kurniawati2022partially} and decentralized MDPs~\cite{bernstein2002complexity,matignon2012coordinated} have been applied to robotics. 
Some works consider risk as simply a cost to minimize, such as~\cite{pereira2013risk}, which considers path planning of autonomous underwater vehicles, and~\cite{primatesta2019risk}, which considers path planning of unmanned aerial vehicles. Other works consider that the risk is constrained while another cost is minimized, like~\cite{feyzabadi2014risk}, which considers robotic path planning in a stochastic environment.
However, while the preceding solutions can account for changes in the environment, including unpredictable ones, they take a one-sided optimization approach, and so they are not suited for adversarial scenarios. 
The work~\cite{aoude2010threat} includes an estimator of the intentions of other, potentially adversarial vehicles but still approaches it as a one-sided path planning problem. 
Therefore, we turn to game theory for results on adversarial interactions.





\begin{figure}
\centering
\includegraphics[width=0.9\linewidth]{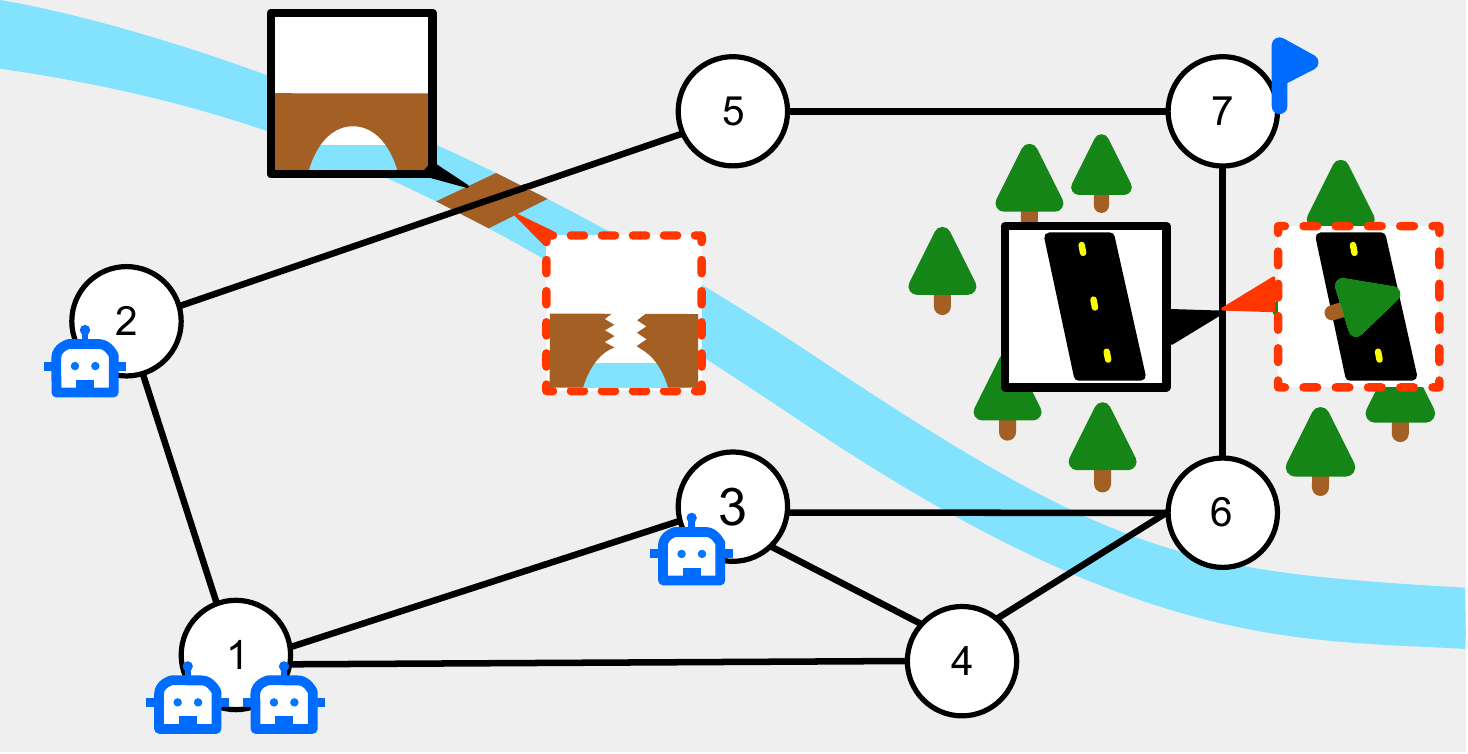}
\caption{An example scenario motivating our problem. The team of robots intend to reach their destination node 7 with minimal risk, 
while an adversary controls the condition of the terrain, with the 
capability to destroy the bridge or down trees on the road. 
} \label{fig:motivatingGraph}
\vspace{-10pt}
\end{figure}

Game theoretic works for robotic path planning include~\cite{emery2005game,wang2020game}. In~\cite{wang2020game}, multiple robots with different objectives interact in a continuous environment and attempt to minimize risk, modeled with higher-order moments of the cost. In~\cite{emery2005game}, the robots are explicitly cooperative but have differing information on a partially observable Markov decision process. 
Closer to our attack/defense scenario, attack graph problems frequently 
model cyber-security
as games on graphs, where the graph represents the possible avenues of attack, which the defenders can influence~\cite{durkota2015optimal, nguyen2017multi}. 
%
For example,~\cite{durkota2015optimal} considers the defender 
modifying the graph nodes to be better defended in a Stackelberg game formulation. Closer to our problem,~\cite{nguyen2017multi} considers feedback heuristic strategies for both players 
to approximate the Nash equilibrium (NE). 
However, in this formulation, the attacker gains permanent control over 
nodes, rather than temporarily occupying them as in our problem.  

Another network security problem over graphs is considered in~\cite{nguyen2009stochastic}, 
where an intrusion detection problem is formulated as a stochastic game~\cite{filar2012competitive}. 
A particular type of game more closely related to our problem is barrier coverage, which considers how a defender may place sensors to detect intruding robots, which attempt to avoid detection~\cite{kloder2008partial,shishika2020game}. 
However, these are not directly applicable to our problem, because they cannot consider real-time feedback in the defender's actions. 
Therefore, this paper's contributions are as follows. 
We formulate the adversarial graph traversal problem as a novel stochastic game, which allows for numerical computation of a mixed NE. We provide theoretical results by bounding the game value with security strategies for both players, and we guarantee that the blue player reaches the goal almost surely under its NE policy. Finally, we demonstrate 
the theoretical results and show the advantage of mixed strategies and coordinated behavior for the robots 
in numerical examples. 

\section{Problem Formulation}





We consider a two-player zero-sum game where the blue player moves its robots through a weighted digraph $\graph{}(t)$ 
which is controlled by the red player. 
Defining $\naturalSet{n} \triangleq \{1,2,\dots n\}$ for $n \in \mathbb{Z}_{>0}$, let $\graph{}(t) = ( \nodeset,\edgeset,\weightmatrix{}(t) )$, where $\nodeset = \naturalSet{N}$ is the set of nodes for $\numbernodes \in \mathbb{Z}_{>0}$, $\edgeset \subset \nodeset \times \nodeset$ is the set of edges, and $\weightmatrix{}(t) \in \mathbb{R}^{\numbernodes \times \numbernodes}$ is the time-varying weighted adjacency matrix,  
where $\weightmatrix{}_{ij}(t) \geq 0$ if $(i,j) \in \edgeset$ and $\weightmatrix{}_{ij}(t) = 0$ otherwise, for $i,j \in \nodeset$. 
This indicates that only the edge weights are time-varying. 
The graph belongs to a known set of graphs, $\graph{}(t) \in \graphset \triangleq \{\graph{1}, \graph{2},\dots,\graph{\numbergraphs} \}$ and, similarly, $\weightmatrix{}(t) \in \weightmatrixset \triangleq\{\weightmatrix{1},\weightmatrix{2},\dots,\weightmatrix{\numbergraphs}\}$, so that $\graph{k} = (\nodeset,\edgeset,\weightmatrix{k})$ for $k \in \naturalSet{K}$, where $\numbergraphs \in \mathbb{Z}_{>0}$ is the number of graphs. 


There are $\numberagents \in \mathbb{Z}_{>0}$ blue robots and each robot $m$ at time $t$ has position $\blueposition{m}(t) \in \nodeset$. 
We collect all robot positions into the vector $\blueposition{}(t) = [p_1,p_2,\dots,p_M]^T \in \nodeset^\numberagents$. 
The set of valid actions for the blue player is $\actionsetblue(\markovstate) = \neighbors{out}{p_1} \times \neighbors{out}{p_2} \times \dots \times \neighbors{out}{p_\numberagents}$, where $\markovstate$ is the game state defined later and $\neighbors{out}{p}$ is the set of out neighbors of node $p$. 
Without loss of generality, node $\numbernodes$ is the goal node, and the game ends when all blue robots reach the goal node ($\blueposition{}(t) = \mathbf{1}_\numberagents \numbernodes$). 
We assume that the goal node $\numbernodes$ is reachable from every other node in the graph and that it has a self-loop with an edge weight of $0$, so that the blue robots can reach the goal from any initial conditions and that any robots which do so can remain there without incurring any additional costs. 
At each time step, based on its action $a_\text{blue}(t) = p(t+1)$, the blue player incurs a stage cost of 
\begin{align}
    C\left(\markovstate,a_\text{blue}(t)\right) = \sum_{m=1}^\numberagents \weightmatrix{}_{\blueposition{m}\blueposition{m}^+}(t) ,
\end{align}
where $\blueposition{m}^+$ is robot $m$'s position at time $t+1$. 

\begin{figure}
\centering
    \includegraphics[width=0.7\linewidth]{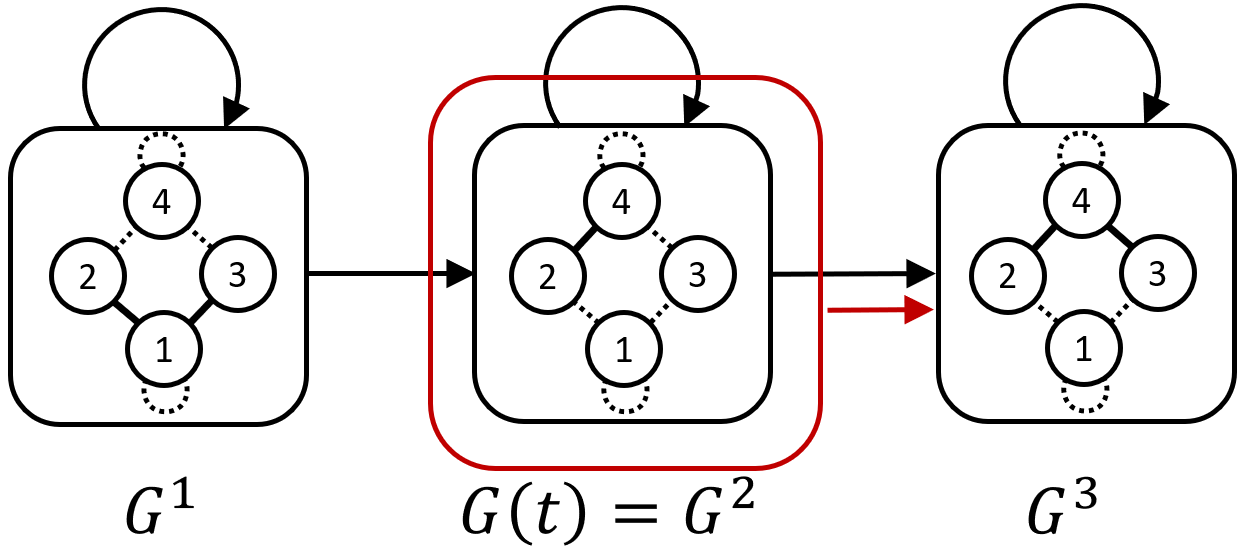}
\caption{This shows an example of red's action graph, where each square node corresponds to a position graph for the blue player, with different edge weights, shown inside. The outlined node indicates $\graph{}(t)$ and the arrow indicates the action of the red player to select $\graph{}(t+1)=\graph{3}$. } \label{fig:redGraph}
\vspace{-10pt}
\end{figure}

The red player's action is to select the graph at the next time step. 
Specifically, its available actions at time $t$ 
are specified by the red player's action graph $\redGraph = (\redNodes,\redEdges)$, where $\redNodes = \naturalSet{K}$ are the nodes of this graph where node $k$ corresponds to the position graph $\graph{k}$, and $\redEdges \subset \redNodes \times \redNodes$ is the set of edges. 
The action graph is unweighted, directed, and weakly connected, and every node has a self loop. 
The red player's action is then 
$a_\text{red}(t) \in \neighbors{red}{k}$, where $\neighbors{red}{k}$ is the set of out neighbors of $k$ in $\redGraph$, and $\graph{k} = \graph{}(t)$. 
Therefore, its available actions are determined by the current graph, to model cases where multiple steps may be required to change the environment and changes may be irreversible. See Figure~\ref{fig:redGraph} for a representation of $\redGraph$. 
Additionally, the red player has a limited amount of ammo $\ammo{t}$, with $\ammo{0} = \maxammo \in \mathbb{Z}_{\geq 0}$, which is consumed to change the graph. 
The red player's set of valid actions is  
\begin{align}
    \actionsetred\big((\blueposition{},\graph{k},\alpha)\big) = \left\{ \begin{matrix}
        \neighbors{red}{k}, & \text{if } \ammo{t} > 0\\
        k , & \text{if } \ammo{t} = 0
    \end{matrix} \right. .
\end{align}

Now, we can write the entire game state, containing the positions of the blue robots, the current graph, and the red player's ammo, as $\markovstate(t) = (\blueposition{}(t),\mathcal{G}(t),\ammo{t}) \in \markovstateset \triangleq \nodeset^\numberagents \times \graphset \times \naturalSet{\maxammo}$. 
The game dynamics is 
\begin{align}\label{eq:gameDynamics}
    \blueposition{}(t+1) &= a_\text{blue}(t) \in \actionsetblue(\markovstate)
    \notag\\
    \graph{}(t+1) &= \graph{a_\text{red}(t)}, a_\text{red}(t) \in \actionsetred(\markovstate)\notag\\
    \ammo{t+1} &= \left\{ \begin{matrix}
        \ammo{t}-1, & \text{for } \graph{a_\text{red}(t)} \neq \graph{}(t)\\
        \ammo{t}, & \text{for } \graph{a_\text{red}(t)} = \graph{}(t)
    \end{matrix} \right. .
\end{align}

To facilitate describing mixed strategies, we formulate a stochastic game, where the state dynamics is a Markov chain and the transition probabilities are influenced by the actions of both players~\cite{filar2012competitive}. The Markov chain is $(\markovstateset,\actionsetblue,\actionsetred,\transitionprob)$  
and dynamics is deterministic, such that the transition probabilities $\transitionprob : \markovstateset \times \markovstateset \times \nodeset^\numberagents \times \naturalSet{\numbergraphs} \rightarrow \{0,1\}$ are in accordance with~\eqref{eq:gameDynamics}. 
%
Mixed policies for the blue player and the red player are defined as mappings $\mixedpolicyblue : \markovstateset \rightarrow \Delta \actionsetblue(\markovstate)$, where $ \Delta \actionsetblue(\markovstate) \subset [0,1]^{|\actionsetblue(\markovstate)|}$ is the simplex over $\actionsetblue(\markovstate)$, and $\mixedpolicyred : \markovstateset \rightarrow \Delta \actionsetred(\markovstate)$, respectively. 
Now, introducing a discount factor $\gamma \in (0,1]$, we can formally write the expected cost of the game as
\begin{align}\label{eq:expectedOutcome}
&\outcome{\markovstate(0),\mixedpolicyblue,\mixedpolicyred}  =  \expectation{\sum_{t=0}^\infty
 \gamma^t C\left(\markovstate(t),a_\text{blue}(t)\right)}{}.  
\end{align}
The zero-cost self-loop on the goal node ensures no further costs are accrued once all blue robots reach the goal. 
Although in our original problem $\gamma=1$, choosing $\gamma \in (0,1)$ 
guarantees that equilibrium policies and a value for the zero-sum stochastic game exist~\cite{shapley1953stochastic,filar2012competitive}. 
Specifically, we seek a pair of strategies $\mixedpolicyblue^*(\cdot),\mixedpolicyred^*(\cdot)$ and a value $V(s)$ which correspond to a NE, so that
\begin{align}\label{eq:NEcondition}
    \outcome{s, \mixedpolicyblue^*, \mixedpolicyred^*} \leq \outcome{s, \mixedpolicyblue, \mixedpolicyred^*}, \notag\\
    \outcome{s, \mixedpolicyblue^*, \mixedpolicyred^*} \geq \outcome{s, \mixedpolicyblue^*, \mixedpolicyred}, \notag\\
    V(s) \triangleq \outcome{s, \mixedpolicyblue^*, \mixedpolicyred^*},  
\end{align}
for all $s \in \markovstateset$ and all valid mixed strategies $\mixedpolicyblue,\mixedpolicyred$. 
In terms of our original problem where $\gamma = 1$, this NE 
is an \emph{approximate} solution, and we will also need to guarantee that the blue robots reach the goal under $\mixedpolicyblue^*$. 

\section{Theoretical Analysis}
We provide some theoretical guarantees before moving on to the main numerical method. Proofs of theorems are omitted for space and the interested reader is referred to~\cite{berneburg2024multi}. 

\subsection{Handling Multiple Blue Robots}\label{sec:multiAgent}


Thus far, we have considered that there are $M \geq 1$ blue robots, but we can consider only a single blue robot, without losing generality. 
Instead of considering multiple robots moving on the graph $\graph{}(t)$, we can equivalently consider a single robot moving on a modified joint state graph $\jointgraph{}(t) = (\jointvertexset,\jointedges,\jointweightmatrix{}(t))$, where each node in $\jointgraph{}(t)$ corresponds to a position vector $\blueposition{}(t)$~\cite{limbu2023team}. 
Therefore, we have a joint graph set $\jointgraphset = \{\jointgraph{1}, \dots, \jointgraph{\numbergraphs} \}$, where $\jointgraph{k} = (\jointvertexset,\jointedges,\jointweightmatrix{k})$ and each joint weighted adjacency matrix $\jointweightmatrix{k} \in \jointweightmatrixset$ corresponds to 
weighted adjacency matrix $\weightmatrix{k}$. 
%
%
To construct this, we assign each position vector $\blueposition{}(t) \in \nodeset^\numberagents$ to a node $\ell$ in $\jointgraph{}$. 
Then, for each node $\ell \in \jointvertexset$, we determine its out-neighbors by finding the next position for each action of the blue player, and assign the edge weight in $\jointweightmatrix{k}$ to be the cost of that action in $\graph{k}$, for each $\graph{k} \in \graphset$. 
Therefore, we are able to convert any game with $\numberagents>1$ on graph $\graph{}(t)$ to a game with $\numberagents=1$ on graph $\jointgraph{}(t)$. 
Finally, we can reduce the size of the joint state graph by considering \emph{indistinguishable} robots. 
Figure~\ref{fig:JSG} shows an example joint state graph for two robots. 

\begin{figure}
    \subfigure[]{
    %
    \includegraphics[width=0.2\linewidth]{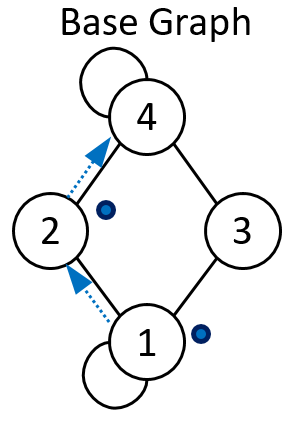}
    %
    }
\subfigure[]{\includegraphics[width=0.6\linewidth]{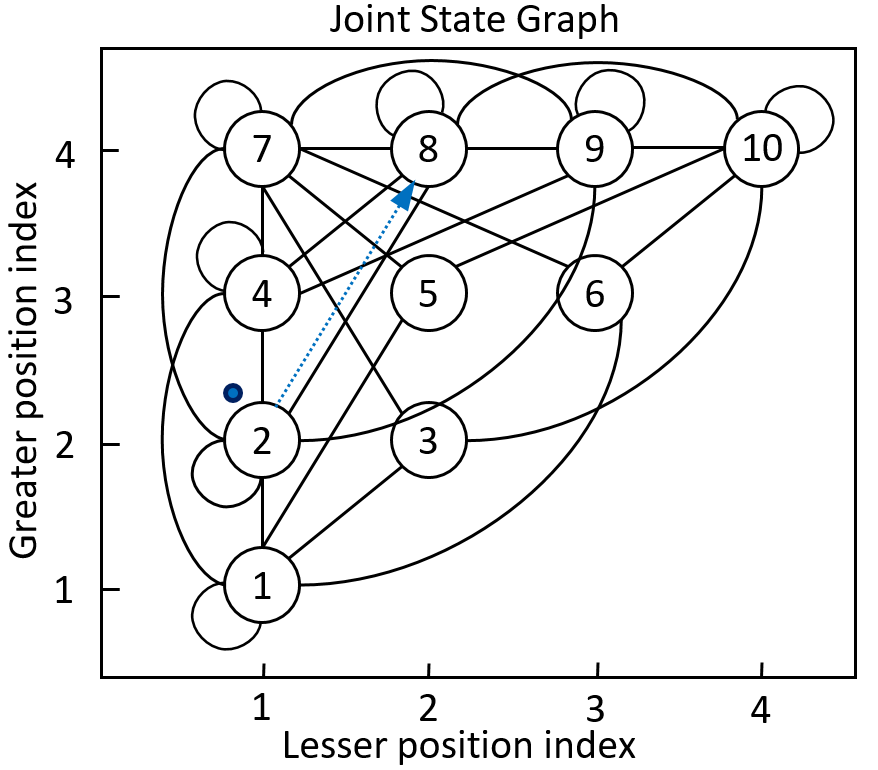}}
\caption{(a) shows a simple example graph, and (b) shows its joint state graph with two robots. The position of each node of the JSG indicates what nodes in the original graph it corresponds to. 
The markers show equivalent positions and the arrows equivalent actions in both graphs. 
} \label{fig:JSG}
\vspace{-10pt}
\end{figure}

\subsection{Game Solution and Value Preliminaries}

Here we provide some preliminaries to aid in analysis. 
Define the ``highest cost'' graph as $\worstcasegraph = (\nodeset,\edgeset,\worstcaseweights)$, where
    $\worstcaseweights_{ij} \triangleq \max_k \weightmatrix{k}_{ij}$ . 
To simplify notation, let $d_{\graph{}}(i,j)$, for $i,j \in \nodeset$ denote the weighted distance from node $i$ to node $j$ on graph $\graph{}$. 
Defining $\edgeset^- \triangleq \edgeset \setminus (N,N)$, the following condition on the discount factor $\gamma$ will be useful later: 
\begin{align}\label{eq:discountFactorCondition}
    \gamma &\geq 1 - \frac{C_\text{min}}{d_\text{max}},\\
    C_\text{min} &\triangleq \min_{(i,j)\in \edgeset^-, k \in \naturalSet{\numbergraphs}} \weightmatrix{k}_{i,j},\quad 
    d_\text{max} \triangleq \max_{p\in \nodeset} d_{\overline{\graph{}}}(p,N). \notag
    \vspace{-25pt}
\end{align}
$C_\text{min}$ is the minimum possible stage cost that can be achieved without $\blueposition{}(t)=N$, which is strictly positive by assumption. 
$d_\text{max}$ is the maximum distance of any node from the goal on the highest cost graph $\overline{\graph{}}$.  
Since 
the right-hand-side of the inequality in~\eqref{eq:discountFactorCondition} is strictly less than $1$,
$\gamma \in (0,1)$ can always be chosen so that~\eqref{eq:discountFactorCondition} is satisfied. 

\subsection{Bounds for the Outcome}

Here we provide security strategies for both players which allow us to bound the expected game outcome. 
These bounds will be simple to compute and can be achieved by simple strategies, compared to finding the game solution. 

\subsubsection{Blue Player's Security Strategy}

Consider a policy for the blue player $\securitypolicyblue$, with a corresponding upper bound $\upperoutcome{s}$ on the expected cost, where the action in state $s\in \markovstateset$ is 
\begin{align}\label{eq:blueSecurityStratAction}
    a_\text{blue} &\in \text{arg} \min_{p^+\in \actionsetblue(s)} \weightmatrix{k}_{p(t)p^+} +d_{\worstcasegraph}(p^+,N) ,\notag\\
    \upperoutcome{s} &= \min_{p^+\in \actionsetblue(s)} \weightmatrix{k}_{p(t)p^+} + d_{\worstcasegraph}(p^+,N). 
\end{align}
This means that the blue player optimizes over the edge costs of the current graph with the distance over the highest cost graph $\worstcasegraph$ as an estimate of the future value. 

\begin{theorem}
    For $\numberagents=1$, if the blue player follows the policy $\securitypolicyblue$ defined by~\eqref{eq:blueSecurityStratAction}, then 
    \begin{align}
        \upperoutcome{s} \geq \outcome{s,\securitypolicyblue,\mixedpolicyred} \text{ and } \upperoutcome{s} \geq V(s), 
    \end{align}
    for all $s \in \markovstateset$ and  $\mixedpolicyred$, with $\upperoutcome{s}$ defined in~\eqref{eq:blueSecurityStratAction}. 
\end{theorem}

\begin{proof}
For reasons of space, we provide only a sketch of the proof. 

If the blue player follows the shortest path on $\worstcasegraph$, then $\outcome{s,\dots}\leq d_{\worstcasegraph}(p,N)$, regardless of the red player's policy. 
Now, if the blue player follows policy $\securitypolicyblue$ at the current time step, but follows the shortest path on $\worstcasegraph$ for all future time, then we have $\outcome{s,\dots}\leq \upperoutcome{s}$. However, we need to show that this holds if the blue player follows policy $\securitypolicyblue$ for all future time. 
Because $\weightmatrix{k}_{p(t)p^+} \leq \worstcaseweights_{p(t)p^+}$, we have $\upperoutcome{s} \leq d_{\worstcasegraph}(p,N)$, trivially from~\eqref{eq:blueSecurityStratAction}. This means that $d_{\worstcasegraph}(p^+,N)$ can act as an upper bound for the future cost, if the blue player uses $\securitypolicyblue$ for all future time. 
%
This allows us to guarantee that the bound is satisfied if the blue player chooses its action according to~\eqref{eq:blueSecurityStratAction}.
%
Finally, to bound the value, note that
\begin{align}
    \upperoutcome{s} \geq \outcome{s,\securitypolicyblue,\mixedpolicyred^*} \geq V(s). 
\end{align}
\end{proof}
This simple strategy enables the blue team to reach the goal with a conservative bound on the cost. 

\subsubsection{Red Player's Security Strategy}
Now, consider a policy $\securitypolicyred$ for the red player, with corresponding lower bound $\loweroutcome{s}$ on the expected cost, where the action in state $s\in \markovstateset$ is 
\begin{align}\label{eq:redSecurityStratAction}
    a_\text{red} \in \text{arg} \max_{k^+ \in \actionsetred(s)} \min_{p^+\in \actionsetblue(s)} \weightmatrix{k}_{pp^+} + \gamma^{N-1} d_\graph{k^+}(p^+,N),\notag\\
    \loweroutcome{s} \triangleq \max_{k^+ \in \actionsetred(s)} \min_{p^+\in \actionsetblue(s)} \weightmatrix{k}_{pp^+} + \gamma^{N-1}d_\graph{k^+}(p^+,N). 
\end{align}
The red player assumes it can only change the graph at the current time and that it will remain constant in the future. 
\begin{theorem}
    Under the condition~\eqref{eq:discountFactorCondition}, for $\numberagents=1$, if the red player follows the policy $\securitypolicyred$ defined by~\eqref{eq:redSecurityStratAction}, then 
    \begin{align}
        \loweroutcome{s}\leq \outcome{s,\mixedpolicyblue,\securitypolicyred} \text{ and } \loweroutcome{s}\leq V(s),
    \end{align}
    for all $s \in \markovstateset$ and $\mixedpolicyblue$, with $\loweroutcome{s}$ defined in~\eqref{eq:redSecurityStratAction}. 
\end{theorem}

\begin{proof}
    Again, for reasons of space, we provide only a sketch of the proof. 
    First, consider the case where the red player never changes the graph, letting $\mixedpolicyred'$ be such that $a_\text{red}(t) = k$ for $\graph{}(t) = \graph{k}$. Now, define
    \begin{align}\label{eq:minDiscountedDistance}
    d^\gamma_{\graph{k}}(p) \triangleq \min_{\mixedpolicyblue} \outcome{(p,\graph{k},\alpha), \mixedpolicyblue, \mixedpolicyred'},
\end{align}
which is the minimum discounted distance to the goal on this graph, and is the result of the blue player's best response to this policy $\mixedpolicyred'$. If condition~\eqref{eq:discountFactorCondition} is satisfied, then it can be shown that shortest paths to the goal do not contain repeated nodes. Therefore, we can upper bound the number of edges traversed on such a path by $\numbernodes-1$. This allows us to claim
\begin{align}
    d^\gamma_{\graph{k}}(p,N) 
    \geq \gamma^{N-2}d_{\graph{k}}(p,N),
\end{align}
because $\gamma^{N-2}$ is the lowest discount which can be achieved on a shortest path. 
Now, consider that the red player selects its action according to~\eqref{eq:redSecurityStratAction} for the current time step and leaves the graph constant for future time steps, in which case the outcome is lower bounded by $\loweroutcome{s}$ for all $\mixedpolicyblue$, by its definition. We must again show that this still holds if the red player follows $\securitypolicyred$ for all future time. 
This follows from the fact that
\begin{align}
    \loweroutcome{(\blueposition,\graph{k},\alpha)} \geq \gamma^{N-2}d_\graph{k}(p,N). 
\end{align}
Finally, to bound the value, note that
\begin{align}
    \loweroutcome{s} \leq \outcome{s,\mixedpolicyblue^*,\securitypolicyred} \leq V(s). 
\end{align}
\end{proof}
This simple strategy provides the red player a conservative bound on its payoff. 

\subsection{Game Solution}

Although the objective is for the robots to reach the goal, 
the introduction of the discount factor $\gamma < 1$ means that the blue player can achieve finite, potentially lower costs without doing so. 
In such cases, solutions to the game are not necessarily solutions to the original problem. 
However, the following condition ensures that, if both players use NE strategies $\mixedpolicyblue^*$ and $\mixedpolicyred^*$, the blue player will reach the goal. 

\begin{theorem}\label{th:solution}
    For $\numberagents=1$, given an initial condition $\markovstate(0) \in \markovstateset$, under the Markov chain with state transitions given by~\eqref{eq:gameDynamics} 
    with actions chosen according to $\mixedpolicyblue^*,\mixedpolicyred^*$ satisfying~\eqref{eq:NEcondition}, if $\gamma$ satisfies~\eqref{eq:discountFactorCondition},
    then $
        \lim_{t\rightarrow \infty } \text{Pr}\left(\blueposition{} (t) = N \right) = 1 $.
\end{theorem}

\begin{proof}
    We present only a sketch of the proof for reasons of space. 
    If the policies of both players are fixed to their NE values, the game reduces to a Markov chain. Therefore, we simply need to show that any state $s = (p,k,\alpha) \in \markovstateset$ where $p = N$ is an absorbing state and any other states are not absorbing for this Markov chain~\cite{kemeny1969finite}. Then, if we can show that the Markov chain is absorbing, we have our desired result. 
    
    
    Because the goal node $N$ has a self-loop with an edge cost of zero, we know that if the blue player reaches the goal, then it accrues no further costs under an equilibrium blue policy, so $\outcome{(N,k,\alpha), \mixedpolicyblue^*, \mixedpolicyred^*} = 0$ and so we must have $\blueposition{}(t)=N$ for all future time and these states are absorbing. 
    
    
    If, under some policy $\mixedpolicyblue$, the blue player never reaches the goal (the probability of reaching the goal from state $\markovstate(t)\in \markovstateset$ is zero), then the value can be lower bounded: 
    \begin{align}
        \outcome{s,\mixedpolicyblue,\mixedpolicyred} \geq \sum_{t=1}^\infty \gamma^t C_\text{min} = \frac{1}{1-\gamma} C_\text{min}, 
    \end{align}
    by assuming that it gets the minimum cost at every time and from evaluating the geometric series.  

    Because we have an upper bound for the value $\upperoutcome{s}$, if we can guarantee that $\upperoutcome{s} < \frac{1}{1-\gamma} C_\text{min}$, then we know that there is a positive probability that, under NE policies, the blue robot reaches the goal at some future time step from state $s$. Noting that $\upperoutcome{s} \leq d_\text{max}$, from the condition~\eqref{eq:discountFactorCondition}, we know that $V(s) \leq \upperoutcome{s} < \frac{1}{1-\gamma} C_\text{min}$ for all $s \in \markovstateset$, so there must be a strictly positive probability that the blue robot reaches the goal from every state $s$. This also guarantees that no state with $p \neq N$ is absorbing. 


    Now, because we have shown that, from every state $\markovstate(t)\in\markovstateset$ there is a strictly positive probability that $p(t') = N$ for some $t' \geq t$, we can claim that the Markov chain induced by NE policies is absorbing. 
    Therefore, 
    we can guarantee that $\lim_{t\rightarrow \infty } \text{Pr}\left(\blueposition{} (t) = N \right) = 1$. 
\end{proof}
Therefore, when $\gamma$ is high enough, the blue team is incentivized to reach the goal.

\section{Numerical Methodology}\label{sec:numMethod}

Because we believe that solving for these equilibrium strategies $\mixedpolicyblue^*, \mixedpolicyred^* $ analytically is not feasible within this work, we turn to numerical methods.
We use a value iteration method based on the work of Shapley~\cite{shapley1953stochastic,filar2012competitive}. 
First, define the Q function $\qmatrix(s) \in \mathbb{R}_{\geq}^{|\actionsetred(s)| \times |\actionsetblue(s)|}$ as a function of the state $s \in \markovstate$, where each element is
\begin{align}
    \qmatrix_{a_r a_b}(s) = C\left(s,a_\text{blue}\right) + \gamma V(s^+),
\end{align}
where $a_r \in \naturalSet{|\actionsetred(s)|}$ and $a_b \in \naturalSet{|\actionsetblue(s)|}$ are action indices corresponding to action $a_\text{red}$ of the red player and $a_\text{blue}$ of the blue player, respectively, 
%
and $s^+$ is the state at the next time step from \eqref{eq:gameDynamics}, given $a_\text{blue}$ and $a_\text{red}$. 
This function gives the expected game outcome if each player takes an arbitrary action at the current time step, and then follows the equilibrium policy 
in the future. 
According to Shapley's Theorem, the equilibrium policies $\mixedpolicyblue^*(s)$, $\mixedpolicyred^*(s)$ and the game value can be found by solving the matrix game defined by $\qmatrix(s)$, for each state $s \in \markovstateset$, where the red row player maximizes and the blue column player  minimizes~\cite{filar2012competitive}. However, because the value function is unknown, 
we use numerical methods to calculate the Q function. 

For each iteration~$\iter$, we compute the estimate of the policies $\mixedpolicyblue^\iter$, $\mixedpolicyred^\iter$ for the given estimate of the Q function $\qfunctionapprox{\iter}$, by solving the matrix game corresponding to $\hat{\qmatrix}^\iter(s)$ for each state $s \in \markovstateset$. 
Next, we update the estimate of the Q function $\qfunctionapprox{\iter+1} $ using those policies as follows:
\begin{align}
    &\hat{\qmatrix}^{\iter+1}_{a_\text{b}a_\text{r}}(s) = C\left(s,a_\text{blue}\right) + \gamma {\mixedpolicyred^\iter}^T \hat{\qmatrix}^\iter(s^+) \mixedpolicyblue^\iter, 
\end{align}
for each $s \in \markovstateset$, $a_b \in \naturalSet{|\actionsetblue(s)|}$, and $a_\text{r} \in \naturalSet{|\actionsetred(s)|}$, and where $s^+$ is found using \eqref{eq:gameDynamics} and $a_\text{blue}$ is the action corresponding to $a_\text{b}$. 
When this converges, $\mixedpolicyblue^\iter,\mixedpolicyred^\iter$ 
can approximate the equilibrium strategies $\mixedpolicyblue^*, \mixedpolicyred^*$.

\subsection{Simplifying Computation} 

Here we provide methods which will simplify the numerical computation of the value. 


\subsubsection{Dominated Blue Actions}

It is possible to determine graph edges which are never beneficial for the blue robot to traverse, by looking for strictly dominated actions. 
For a given game state $\markovstate(t) \in \markovstateset$, let $p(t) = p_0 \in \nodeset$, and consider two actions of the blue player, $p_1,p_2 \in \actionsetblue(\markovstate)$. A sufficient condition for action $p_2$ to dominate action $p_1$ is
\begin{align}\label{eq:conservativeDominationCondition}
    \stagecost{\markovstate,p_1} + \gamma\loweroutcome{(p_1,\graph{a_\text{r}},\alpha_{a_\text{r}}) } > \notag\\
    \stagecost{\markovstate,p_2} + \gamma\upperoutcome{(p_2,\graph{a_\text{r}},\alpha_{a_\text{r}}) },
\end{align}
for all $a_\text{r} \in \actionsetred(\markovstate)$, where $\alpha_{a_\text{r}}$ is the red player's new ammo count as determined by $\markovstate$ and $a_\text{r}$ according to~\eqref{eq:gameDynamics}. 
Because the blue player never benefits from using a dominated strategy, 
if condition~\eqref{eq:conservativeDominationCondition} is satisfied for some $p_0,p_1,p_2$ and all $\alpha,\graph{}(t)$, then edge $(p_0,p_1)$ can be removed from the graphs, without changing the results. Therefore, one could iterate over all pairs of actions for all nodes, removing the edges corresponding to dominated actions and any nodes from which the goal cannot be reached. 
This allows the game to be solved for fewer states and actions for the blue player, without changing the results. 

\subsubsection{Sub-Game Formulation}

From equations~\eqref{eq:expectedOutcome} and~\eqref{eq:NEcondition}, we can write the value as
\begin{align}
V(s) 
&= \sum_{s' \in \markovstateset} \text{Pr}\Bigl\{s^+=s' 
\Bigr\} 
\Bigl( \stagecost{s,p'} + \gamma V(s') \Bigr) ,\notag
\end{align}
for a given game state $\markovstate(t) = s = (\blueposition{},\graph{},\alpha) \in \markovstateset$, where $s^+ = (\blueposition{}^+,\graph{}^+,\alpha^+) = \markovstate(t+1) \in \markovstateset$ and $s' = (p', \graph{}', \alpha')\in \markovstateset$
, and $\text{Pr}\{s(t+1)=s'\}$ is found according to~\eqref{eq:gameDynamics} when the actions are chosen according to NE policies. 
Now, 
either the graph remains the same, in which case $s^+ = (\blueposition{}^+,\graph{},\alpha)$, or it is changed by the red player, in which case $s^+ = (\blueposition{}^+,\graph{}^+,\alpha-1)$. This implies that the value of the game at some ammo state $\alpha$ depends only on the value at lower ammo states $\alpha' \leq \alpha$. We can take advantage of this to formulate sub-games which allow us to solve the original game. 

Instead of considering the graph $\graph{}$ and the ammo $\alpha$ to be part of the game state, we consider them to be parameters of a sub-game with game state
$\blueposition{}$. 
The sub-game ends when red takes action to change the graph. 
The actions of the players, the game dynamics, and the stage costs are the same as before, but now we consider a terminal cost when the graph changes. 
The value for the sub-game defined by graph $\graph{}$ and ammo $\alpha$ 
is $V_{\graph{},\alpha}(\blueposition{})$, 
and the terminal costs will be $V_{\graph{}(T),\alpha-1}(\blueposition{T})$, where $\blueposition{T}$ is blue's position when red takes action to change the graph to $\graph{}(T)\neq \graph{}$ at time $T-1$. 
The outcome of this sub-game, for $\blueposition{}(0)=p$, is 
\begin{align*}
&V_{\graph{},\alpha}(\blueposition{}) = \\
&E\Big[\sum_{t=0}^{T-1} \gamma^tC\left(\markovstate(t),\blueposition{}(t+1)\right) 
+ \gamma^T V_{\graph{}(T),\alpha-1}(\blueposition{}(T))\Big]. 
\end{align*}
Because $V_{\graph{},0}(\blueposition{}) = d_\graph{}(\blueposition{},N)$, 
the value of the full game can be found by solving the sub-games, using the value of the sub-games with ammo $\alpha$ to solve for the value of the sub-games with ammo $\alpha+1$. 
Compared to the full game with $\numbernodes \numbergraphs \maxammo$ states, 
for $\numbergraphs$ graphs of $\numbernodes$ nodes and $\maxammo$ ammo counts, the trade-off is considering $\numbergraphs \maxammo$ sub-games with $\numbernodes$ states each. 
The latter may be more computationally efficient, because the number of game states is the number of  Q-matrices which must be updated and the number of matrix games which must be solved at every iteration. 

\section{Numerical Results}

Here we show the results of applying the numerical methods in Section~\ref{sec:numMethod}. The discount factor was chosen to be $\gamma = 1-10^{-9} \approx 1$ and the red player's action graph~$\redGraph$ is complete. First, we examine results for some examples, and then we perform a statistical analysis of randomly constructed graphs. 

\subsection{Illustrative Examples}

\paragraph{Benefit of Mixing} 
Figure~\ref{fig:graphBranchLoopMix3} shows an example graph where each player mixes between three options from node $1$, with the blue player mixing between remaining at that node and between taking either branch. 
With $s(0) = (1,\graph{1},1)$, 
the blue player chooses its action according to $[\text{Pr}(p^+=1),\text{Pr}(p^+=2),\text{Pr}(p^+=3)] = [0.5,0.25,0.25]$, while the red player chooses its action according to $\mixedpolicyred^*(s(0)) = [0.5,0.25,0.25]^T$. 
The blue player seems to mix between remaining at the current node and moving towards the goal when the current graph has a high immediate edge cost, while the red player is incentivized to change the graph because the current graph has lower costs closer to the goal. 
On the other hand, if $\graph{}(0)=\graph{2}$ instead, the blue player chooses its action according to $[\text{Pr}(p^+=1),\text{Pr}(p^+=2),\text{Pr}(p^+=3)] = [0, 0.5, 0.5]$,
and the red player chooses its action according to $\mixedpolicyred^*(s(0)) = [0,0.75,0.25]^T$. 
The players each only mix between two actions, so that the blue player attempts to take the cheaper path, which the red player attempts to prevent. 

\begin{figure}
\includegraphics[width=\linewidth]{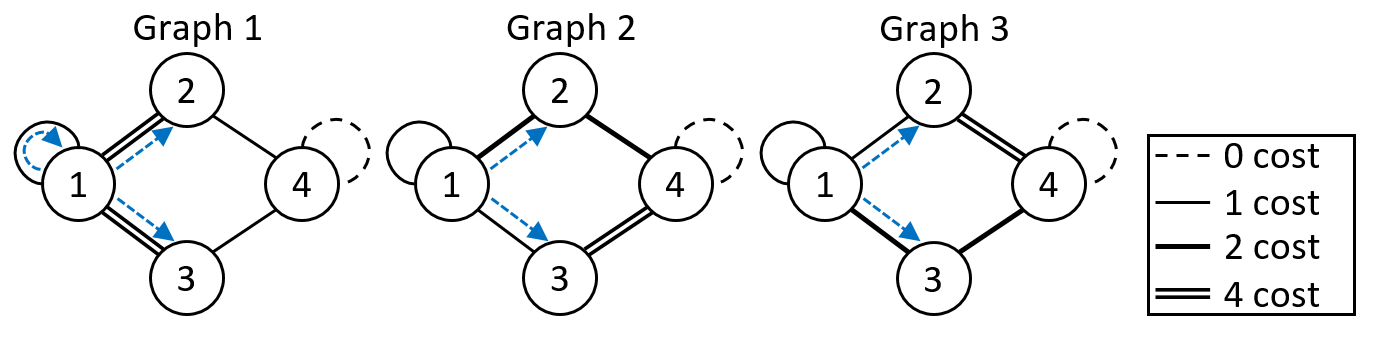}
\caption{This demonstrates mixed equilibrium policies of both players. 
The arrows indicate the possible actions of the blue player under that policy from node $1$. 
On graph $1$, the red player also mixes over its three options, but on graph $2$ or $3$, it only mixes between selecting graph $2$ and graph $3$. 
} 
\label{fig:graphBranchLoopMix3}
\vspace{-10pt}
\end{figure}

\paragraph{Benefit of Multi-Robot Coordination} 
The graph in Figure~\ref{fig:graphMultiAgents} shows a case where having more robots is beneficial for the blue player. 
Due to the longer length of these branches (compared to the graphs in Figure~\ref{fig:graphBranchLoopMix3}), the red player can simply choose $\graph{}(t+1) = \graph{1}$ when $p(t) = 4$ or $\graph{}(t+1) = \graph{2}$ when $p(t) = 3$, forcing the blue player to take the higher edge cost. 
On the other hand, multiple robots can split up in a coordinated fashion, so that at least one of them takes the cheaper path. To formalize this benefit, we define the average value per robot as $\valueperagent{M}(s)\triangleq V(s)/\numberagents$. When $\ammo{0}\geq 1$, for a single robot ($\numberagents = 1$), 
$\valueperagent{1}(\markovstate(0)) = 20$ with $p(0) = 1$, but for two robots, $\valueperagent{2}(\markovstate(0)) = 13$ with $p(0) = [1,1]^T$, 
and $\ammo{0}\geq 1$. This indicates that, depending on the graphs, the blue player can gain an advantage from having multiple robots 
which split up in a coordinated fashion. 
Another coordinated behavior seen in Figure~\ref{fig:graphMultiAgents} is that the robot which could traverse fewer edges to the goal instead waits at node $5$. 
If this leading robot instead went to the goal immediately, then the red player could switch the graph so that both robots must take the $16$ cost to the goal. Instead, waiting ensures that both robots traverse the final risky edges synchronously and exactly one of them achieves the cheaper cost.

\begin{figure}
\includegraphics[width=\linewidth]{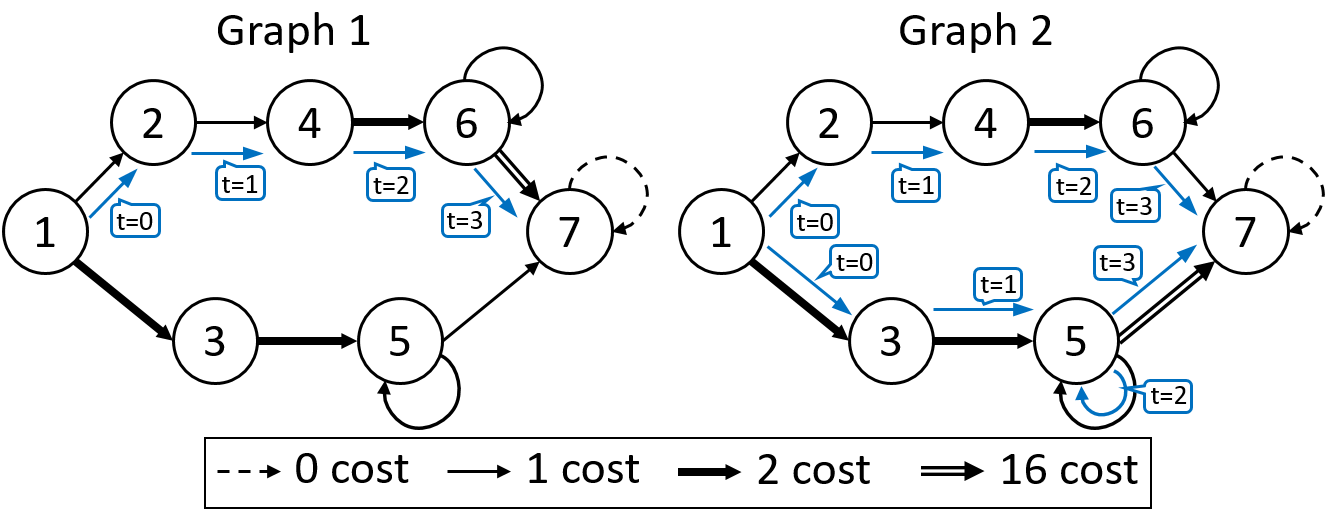}
\caption{Here, it is better for the blue player to have multiple robots which split up to reach the goal node $7$. The blue arrows on graph $1$ show the policy of a single robot, with the speech bubbles indicating the time of each action. The blue arrows on graph $2$ indicate the paths followed by two robots. For $M=1$, and $\ammo{0} = 1$, the red player will always change the graph 
so that the blue player will get a cost of $16$ for taking the $(6,7)$ edge, but for two robots exactly one robot will take the cheaper path to the goal, no matter the red player's action. } \label{fig:graphMultiAgents}
\vspace{-10pt}
\end{figure}

\paragraph{A More Practical Example}

Here, we discuss how these behaviors fit into a larger graph. Figure~\ref{fig:bigGameSolution} shows portions of an example trajectory under equilibrium policies for a larger graph of $\numbernodes=10$ nodes, with $\numberagents=4$ blue robots, and a maximum ammo count for the red player of $\maxammo = 5$, while Table~\ref{tab:bigGameSolutionDetails} contains more details. Both players employ mixed strategies, so this is only one of the possible outcomes. 
%
For example, at $t=1$, the blue player mixes so that its next possible positions are $[1,1,5,6]^T$, $[1,1,6,6]^T$, and $[1,2,6,6]^T$, so that some robots wait and some robots may split up, while the red player's policy is mixed between choosing the each of the three graphs. 
The value per robot is $\valueperagent{4}(\markovstate(0)) = 14.4667$, while it is higher for 
a single robot at $\valueperagent{1}(\markovstate(0)) = 15.2474$, again suggesting that having more robots is helpful for the blue player. Finally, this value is between the lower and upper bounds for the value, $\loweroutcome{\markovstate(0)} = 12$ and $\upperoutcome{\markovstate(0)} = 80$, as we expect.

\begin{figure*}
\includegraphics[width=\linewidth]{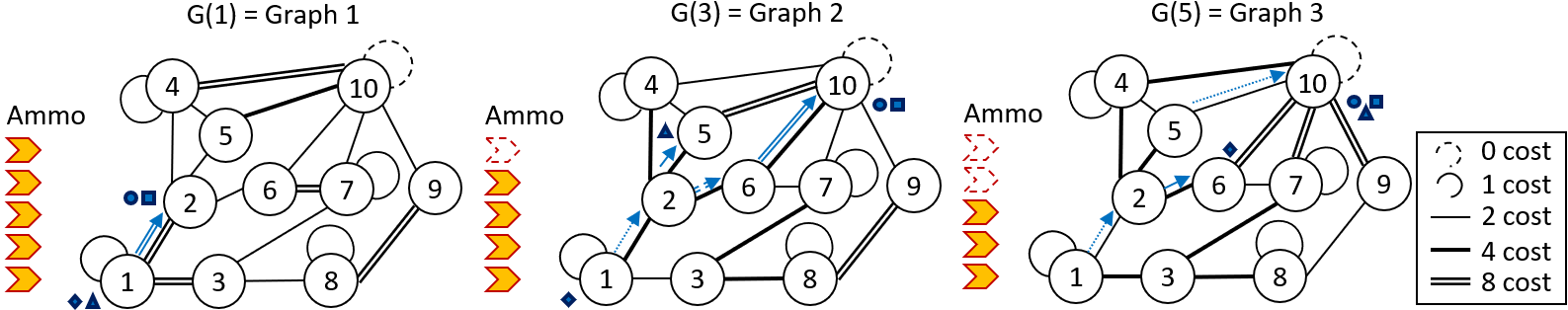}
\caption{This shows part of an example trajectory, under equilibrium strategies, for a game on a larger graph with the goal as node 10, with four blue robots. 
The game state is shown for $t \in \{1,3,5\}$. 
Each blue shape indicates the position of a blue robot, 
the filled chevrons indicate the remaining ammo, 
and the arrows indicate the blue robots' movement on previous time steps. Specifically, the double-lined arrows indicate two robots moving together, the solid lines indicate movement on the preceding time step, and the dotted lines indicate movement on the time step before that. 
See Table~\ref{tab:bigGameSolutionDetails} for the full trajectory. 
We also observe that the blue robots wait at node 1 and split up, so that not all take the expensive first edges from node 1 on graph 1. 
} \label{fig:bigGameSolution}
\end{figure*}

\begin{table}[htb]
    \centering
    \begin{tabular}{|c|c|c|c|c|}
        \hline
         $t$ & $p(t)$ & $\graph{}(t)$ & $\ammo{t}$  & $\sum_{t'=0}^{t-1}C(t')$ \\
        \hline
        0 & $[1,1,1,1]^T$ & $\graph{1}$ & $5$ & 0 \\
        \hline
        1 & $[2,2,1,1]^T$ & $\graph{1}$ & $5$ & 18 \\
        \hline
        2 & $[6,6,2,1]^T$ & $\graph{1}$ & $5$  & 31 \\
        \hline
        3 & $[10,10,5,1]^T$ & $\graph{2}$ & $4$ & 38 \\
        \hline
        4 & $[10,10,10,2]^T$ & $\graph{3}$ & $3$ & 50 \\
        \hline
        5 & $[10,10,10,6]^T$ & $\graph{3}$ & $3$ & 54 \\
        \hline
        6 & $[10,10,10,10]^T$ & $\graph{3}$ & $3$ & 62 \\
        \hline
     \end{tabular}
    \caption{
    Example NE trajectory 
    for Figure~\ref{fig:bigGameSolution} 
    }
    \label{tab:bigGameSolutionDetails}
    \vspace{-20pt}
\end{table}

\subsection{Statistical Results}

To statistically demonstrate some properties of this game and its solution, we constructed graphs using a modification of the canonical Erdős–Rényi model~\cite{erdHos1960evolution}. For a set of nodes $\nodeset$ of size $\numbernodes_\text{max}$, the probability that any pair of nodes $(i,j) \in \nodeset \times \nodeset$ 
is in the directed edge set of the graph is a chosen constant~$0.5$. 
Node $1$ and the goal node are the nodes with the longest (unweighted) distance between them, and any nodes from which the goal node cannot be reached are pruned, so the graph size is $\numbernodes \leq \numbernodes_\text{max}$. 
Finally, a self-loop is added to the goal node $\numbernodes$. 
We set $\numbergraphs = 3$, and randomly assign edge weights for each graph~$\graph{k}$ so that, for 
$(i,j) \in \edgeset$ and $i \neq j$, $(\weightmatrix{1}_{ij}, \weightmatrix{2}_{ij}, \weightmatrix{3}_{ij})$ is a permutation of the set $\{2,4,8\}$, chosen with uniform probability. 
The self-loop weights are all $1$ except for the goal node's. 
Using this method, $100$ graphs were generated for each initial graph size $\numbernodes_\text{max}$ between 4 nodes and 8 nodes. 

The common initial condition is 
$p(0)=\mathbf{1}_M$ and $\graph{}(0) = \graph{1}$, and the red player has ammo $\ammo{0}=6$. 
In Figure~\ref{fig:statsCostByNodes}, the values for this initial condition, for 5, 6, and 7 node graphs, have been normalized so that the lower bound~$\loweroutcome{s(0)}$ is $0$ and the upper bound~$\upperoutcome{s(0)}$ is $1$ in the plot. Trivial graphs where $\loweroutcome{s(0)} = \upperoutcome{s(0)}$ have been excluded from the results. 
Additionally, the approximate expected cost is shown for alternative blue strategies, which are the security strategy defined in~\eqref{eq:blueSecurityStratAction}, and a naive strategy where robots follow the shortest path on a ``best case'' graph, where all edges have the smallest cost. 
These have the same expected cost for any number of robots. 
The bounds are respected and are sometimes tight for all NE policies and the security strategy, while the naive strategy sometimes 
achieves much higher normalized costs for the blue player. 
Compared to the baseline of the security strategy, the blue NE policy has lower costs, and both outperform the naive blue strategy. 
Additionally, the mean \emph{normalized} value exhibits a downward trend as the number of nodes increases. 
More nodes may give the blue player more options to reach the goal, 
allowing it to gain an advantage through mixed policies. 
%
\begin{figure}
\includegraphics[width=\linewidth]{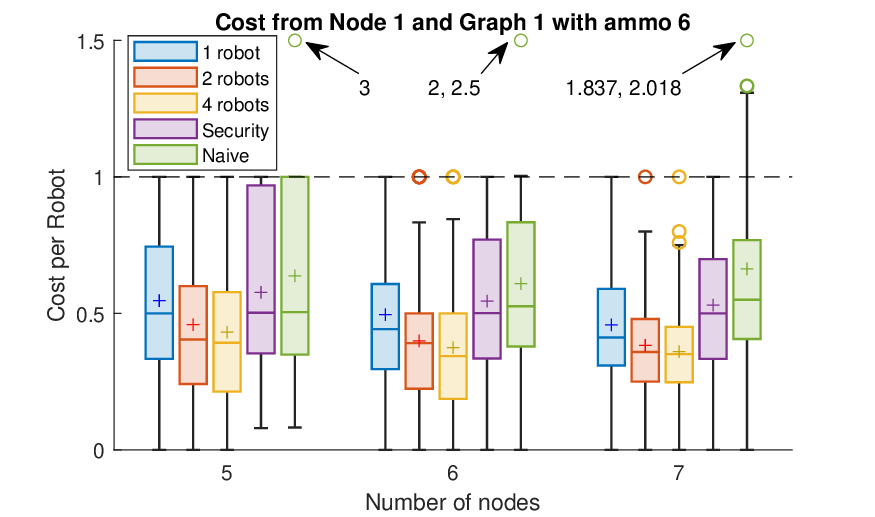}
\caption{This shows the value per robot $\valueperagent{M}(\markovstate(0))$ 
over 
random graphs, 
which has been normalized so that the lower bound is $0$ and the upper bound is $1$. 
For each number of blue robots, and for the blue security strategy and naive strategy, the value is plotted against the number of nodes in the random graphs. The plus sign indicates the mean, the center horizontal line indicates the median value, 
the middle $50\%$ of the values are inside the box, 
the whiskers indicate the furthest non-outlier values, and the circles indicate the outliers. 
Some outliers beyond the axis bounds have their values shown in annotations. 
The blue player's NE policy outperforms the baseline security strategy, which in turn outperforms the naive strategy, which cannot guarantee the upper bound on the cost. 
The values are always between the upper and lower bounds for the NE policies and the security policy, the mean decreases as the number of nodes increases, 
and the mean decreases with the number of robots.
} \label{fig:statsCostByNodes}
\end{figure}
%
%
%
Figure~\ref{fig:statsCostByNodes} also plots the values per robot $\valueperagent{M}(\markovstate(0))$ for the random graphs for different numbers of robots. 
The mean value per robot decreases as number of robots increases, for each graph size, demonstrating that having more robots is an advantage for the blue player under NE policies. 

Figure~\ref{fig:statsAmmo} shows the effect of the red player's ammo count on the normalized game value. 
The average of the normalized values is taken over all the random graphs and plotted against the ammo count. 
The value is nondecreasing with respect to the ammo count, and the value plateaus around 4 ammo, below the upper bound for the value, suggesting that having more ammo stops benefiting the red player after a point. 

\begin{figure}
\includegraphics[width=\linewidth]{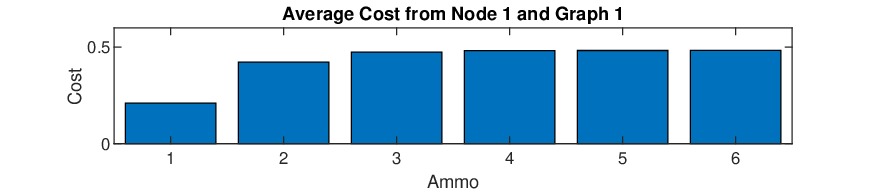}
\caption{This shows the average game value, over about $500$ random graphs as a function of ammo, normalized so that the lower bound is $0$ and the upper bound is $1$. 
The value is nondecreasing in ammo and it asymptotically approaches a value which is significantly lower than the upper bound. 
} \label{fig:statsAmmo}
\end{figure}

\section{Conclusions}

To handle uncertainty for robots moving through hazardous, time-varying environments in the presence of an adversary, we have formulated a novel zero-sum game where the blue player's robots move through a graph with time-varying costs to reach a goal, while the red player controls those costs. This can be appropriate for military or security applications, or to model the worst-case natural changes. 
Due to the complexity of the game, we provided security strategies for both players and a theoretical guarantee that the blue robots reach the goal. We provided methods to simplify computation, and demonstrated the game solution through numerical methods. Both players benefit from mixed strategies, while the blue player benefits from having multiple robots which use coordinated strategies, such as splitting up and waiting to traverse risky edges synchronously, ensuring at least some take the cheaper paths. 
While this work made the strong assumption of perfect \emph{state} information for the blue player, future work can include formulating and solving games of incomplete information, where the blue player is uncertain of the red player's capabilities, and imperfect information, where the blue robots are only aware of changes in the graph close to them, to model the more realistic case of limited perception for the robots.

{\scriptsize

    \bibliographystyle{ieeetr}
    \bibliography{references}

\begin{thebibliography}{10}

\bibitem{gasparetto2015path}
A.~Gasparetto, P.~Boscariol, A.~Lanzutti, and R.~Vidoni, ``Path planning and trajectory planning algorithms: A general overview,'' {\em Motion and operation planning of robotic systems: Background and practical approaches}, pp.~3--27, 2015.

\bibitem{hsu2002randomized}
D.~Hsu, R.~Kindel, J.-C. Latombe, and S.~Rock, ``Randomized kinodynamic motion planning with moving obstacles,'' {\em The International Journal of Robotics Research}, vol.~21, no.~3, pp.~233--255, 2002.

\bibitem{ding2008finding}
B.~Ding, J.~X. Yu, and L.~Qin, ``Finding time-dependent shortest paths over large graphs,'' in {\em Proceedings of the 11th international conference on Extending database technology: Advances in database technology}, pp.~205--216, 2008.

\bibitem{cai1997time}
X.~Cai, T.~Kloks, and C.-K. Wong, ``Time-varying shortest path problems with constraints,'' {\em Networks: An International Journal}, vol.~29, no.~3, pp.~141--150, 1997.

\bibitem{yuan2019constrained}
Y.~Yuan, X.~Lian, G.~Wang, Y.~Ma, and Y.~Wang, ``Constrained shortest path query in a large time-dependent graph,'' {\em Proceedings of the VLDB Endowment}, vol.~12, no.~10, pp.~1058--1070, 2019.

\bibitem{dimmig2023multi}
C.~A. Dimmig, K.~C. Wolfe, and J.~Moore, ``Multi-robot planning on dynamic topological graphs using mixed-integer programming,'' in {\em 2023 IEEE/RSJ International Conference on Intelligent Robots and Systems (IROS)}, pp.~5394--5401, IEEE, 2023.

\bibitem{puterman1990markov}
M.~L. Puterman, ``Markov decision processes,'' {\em Handbooks in operations research and management science}, vol.~2, pp.~331--434, 1990.

\bibitem{kurniawati2022partially}
H.~Kurniawati, ``Partially observable markov decision processes and robotics,'' {\em Annual Review of Control, Robotics, and Autonomous Systems}, vol.~5, no.~1, pp.~253--277, 2022.

\bibitem{bernstein2002complexity}
D.~S. Bernstein, R.~Givan, N.~Immerman, and S.~Zilberstein, ``The complexity of decentralized control of markov decision processes,'' {\em Mathematics of operations research}, vol.~27, no.~4, pp.~819--840, 2002.

\bibitem{matignon2012coordinated}
L.~Matignon, L.~Jeanpierre, and A.-I. Mouaddib, ``Coordinated multi-robot exploration under communication constraints using decentralized markov decision processes,'' in {\em Proceedings of the AAAI Conference on Artificial Intelligence}, vol.~26, pp.~2017--2023, 2012.

\bibitem{pereira2013risk}
A.~A. Pereira, J.~Binney, G.~A. Hollinger, and G.~S. Sukhatme, ``Risk-aware path planning for autonomous underwater vehicles using predictive ocean models,'' {\em Journal of Field Robotics}, vol.~30, no.~5, pp.~741--762, 2013.

\bibitem{primatesta2019risk}
S.~Primatesta, G.~Guglieri, and A.~Rizzo, ``A risk-aware path planning strategy for uavs in urban environments,'' {\em Journal of Intelligent \& Robotic Systems}, vol.~95, pp.~629--643, 2019.

\bibitem{feyzabadi2014risk}
S.~Feyzabadi and S.~Carpin, ``Risk-aware path planning using hirerachical constrained markov decision processes,'' in {\em 2014 IEEE International Conference on Automation Science and Engineering (CASE)}, pp.~297--303, IEEE, 2014.

\bibitem{aoude2010threat}
G.~S. Aoude, B.~D. Luders, D.~S. Levine, and J.~P. How, ``Threat-aware path planning in uncertain urban environments,'' in {\em 2010 IEEE/RSJ International Conference on Intelligent Robots and Systems}, pp.~6058--6063, IEEE, 2010.

\bibitem{emery2005game}
R.~Emery-Montemerlo, G.~Gordon, J.~Schneider, and S.~Thrun, ``Game theoretic control for robot teams,'' in {\em Proceedings of the 2005 IEEE International Conference on Robotics and Automation}, pp.~1163--1169, IEEE, 2005.

\bibitem{wang2020game}
M.~Wang, N.~Mehr, A.~Gaidon, and M.~Schwager, ``Game-theoretic planning for risk-aware interactive agents,'' in {\em 2020 IEEE/RSJ International Conference on Intelligent Robots and Systems (IROS)}, pp.~6998--7005, IEEE, 2020.

\bibitem{durkota2015optimal}
K.~Durkota, V.~Lisy, B.~Bo{\v{s}}ansky, and C.~Kiekintveld, ``Optimal network security hardening using attack graph games,'' in {\em Proceedings of IJCAI}, pp.~7--14, 2015.

\bibitem{nguyen2017multi}
T.~H. Nguyen, M.~Wright, M.~P. Wellman, and S.~Baveja, ``Multi-stage attack graph security games: Heuristic strategies, with empirical game-theoretic analysis,'' in {\em Proceedings of the 2017 Workshop on Moving Target Defense}, pp.~87--97, 2017.

\bibitem{nguyen2009stochastic}
K.~C. Nguyen, T.~Alpcan, and T.~Basar, ``Stochastic games for security in networks with interdependent nodes,'' in {\em 2009 International Conference on Game Theory for Networks}, pp.~697--703, IEEE, 2009.

\bibitem{filar2012competitive}
J.~Filar and K.~Vrieze, {\em Competitive Markov decision processes}.
\newblock Springer Science \& Business Media, 2012.

\bibitem{kloder2008partial}
S.~Kloder and S.~Hutchinson, ``Partial barrier coverage: Using game theory to optimize probability of undetected intrusion in polygonal environments,'' in {\em 2008 IEEE International Conference on Robotics and Automation}, pp.~2671--2676, IEEE, 2008.

\bibitem{shishika2020game}
D.~Shishika, D.~G. Macharet, B.~M. Sadler, and V.~Kumar, ``Game theoretic formation design for probabilistic barrier coverage,'' in {\em 2020 IEEE/RSJ International Conference on Intelligent Robots and Systems (IROS)}, pp.~11703--11709, IEEE, 2020.

\bibitem{shapley1953stochastic}
L.~S. Shapley, ``Stochastic games,'' {\em Proceedings of the national academy of sciences}, vol.~39, no.~10, pp.~1095--1100, 1953.

\bibitem{berneburg2024multi}
J.~Berneburg, X.~Wang, X.~Xiao, and D.~Shishika, ``Multi-robot coordination induced in hazardous environments through an adversarial graph-traversal game,'' {\em arXiv preprint arXiv:2409.08222}, 2024.

\bibitem{limbu2023team}
M.~Limbu, Z.~Hu, S.~Oughourli, X.~Wang, X.~Xiao, and D.~Shishika, ``Team coordination on graphs with state-dependent edge costs,'' in {\em 2023 IEEE/RSJ International Conference on Intelligent Robots and Systems (IROS)}, pp.~679--684, IEEE, 2023.

\bibitem{kemeny1969finite}
J.~G. Kemeny and J.~L. Snell, {\em Finite markov chains}, vol.~26.
\newblock van Nostrand Princeton, NJ, 1969.

\bibitem{erdHos1960evolution}
P.~Erd{\H{o}}s, A.~R{\'e}nyi, {\em et~al.}, ``On the evolution of random graphs,'' {\em Publ. math. inst. hung. acad. sci}, vol.~5, no.~1, pp.~17--60, 1960.

\end{thebibliography}

}

\end{document}